\documentclass[10pt]{article}
\usepackage{amsmath,amsthm}
\newtheorem{theorem}{Theorem}[section]
\newtheorem{lemma}[theorem]{Lemma}
\newtheorem{conjecture}[theorem]{Conjecture}
\newtheorem{example}{Example}[section]
\newcommand{\LP}{\ensuremath{\mathit{LP}}}
\newcommand{\LPF}{\ensuremath{\mathit{LP_f}}}
\newcommand{\NF}{\ensuremath{\mathit{N_f}}}

\begin{document}

\title{\bf Putting Dots in Triangles}
\author{\Large Simon~R.~Blackburn
%\thanks{Research supported by EPSRC grant EP/D053285/1}
\\Royal Holloway, University of
London\\Egham, Surrey TW20 0TN, United Kingdom 
\and \Large Maura~B.~Paterson
%\thanks{Research supported by EPSRC grant EP/D053285/1}
\\Birkbeck College, University of
London\\Malet Street, London WC1E 7HX, United Kingdom \and
\Large Douglas~R.~Stinson
%\thanks{Research supported by NSERC grant 203114-06}
\\David R. Cheriton School of Computer Science\\University of
Waterloo, Waterloo, ON, N2L 3G1, Canada}
\date{}
\maketitle
\thispagestyle{empty}
\pagestyle{empty}

\begin{abstract}
Given a right-angled triangle of squares in a grid whose horizontal and vertical sides are $n$ squares long, let $N(n)$ denote the maximum number of dots that can be placed into the cells of the triangle such that each row, each column, and each diagonal parallel to the long side of the triangle contains at most one dot. It has been proven that $\NF(n) = \left\lfloor \frac{2n+1}{3} \right\rfloor$. In this note, we give a new proof of the upper bound 
$\NF(n) \leq \left\lfloor \frac{2n+1}{3} \right\rfloor$ using linear programming techniques.
\end{abstract}

\section{The problem}
Consider a `triangle' of squares in a grid whose sides are $n$ squares long, as illustrated by the diagram in Figure \ref{7.fig}, for which $n=7$.
\begin{figure}
\caption{A triangle of size 7}
\label{7.fig}
\begin{equation*}
\setlength{\unitlength}{.5mm}
\begin{picture}(100,90)
\put(10,10){\line(1,0){70}}
\put(10,20){\line(1,0){70}}
\put(20,30){\line(1,0){60}}
\put(30,40){\line(1,0){50}}
\put(40,50){\line(1,0){40}}
\put(50,60){\line(1,0){30}}
\put(60,70){\line(1,0){20}}
\put(70,80){\line(1,0){10}}

\put(10,10){\line(0,1){10}}
\put(20,10){\line(0,1){20}}
\put(30,10){\line(0,1){30}}
\put(40,10){\line(0,1){40}}
\put(50,10){\line(0,1){50}}
\put(60,10){\line(0,1){60}}
\put(70,10){\line(0,1){70}}
\put(80,10){\line(0,1){70}}

\put(40,2){\makebox(10,5){$\longleftarrow$ $n$ $\longrightarrow$ }}
\end{picture}
\end{equation*}
\end{figure}
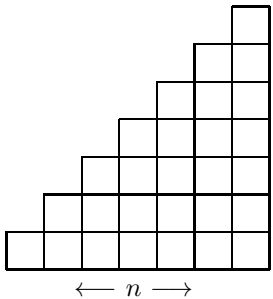
We call a southwest-to-northeast diagonal a {\it standard diagonal}.
Note that our triangle consists of all the cells in an $n \times n$ square
that lie on or below the longest standard diagonal.

We denote by $N(n)$ the maximum number of dots that can be placed into the cells of the triangle such that each row, each column, and each standard diagonal
contains at most one dot.
Determining $N(n)$ is equivalent to solving the following problem:
Suppose we have a chessboard made up of hexagonal cells arranged in
the shape of an equilateral triangle of side $n$. Then $N(n)$ is the
maximum number of non-attacking queens that can be placed on such a
board, where a queen can move in any one of the three directions allowed
on a hexagonal grid.
The following theorem was proven by Vaderlind, Guy and Larson 
\cite[Problem 252]{VGL} and independently by Nivasch and Lev \cite{NL}:

\begin{theorem} 
\label{main.thm}
$N(n) = \NF(n)$, where
\begin{eqnarray*}
\NF(3t) &=& 2t \\
\NF(3t+1) &=& 2t+1 \\
\NF(3t+2) &=& 2t+1. 
\end{eqnarray*}
\end{theorem}

Note that the value of $\NF(n)$ can be stated more succinctly as follows:
\[ \NF(n) = \left\lfloor \frac{2n+1}{3} \right\rfloor .\]

In order to prove Theorem \ref{main.thm}, we require a construction to 
establish the lower bound $N(n) \geq \NF(n)$ as well as a proof of the upper bound
$N(n) \leq \NF(n)$. In \cite{NL,VGL}, the upper bound was proven
by elementary combinatorial arguments. In this note, we give a new proof of the upper bound
using linear programming techniques. In the end, our proof is also combinatorial;
the main contribution we make is to demonstrate the use of linear programming techniques
in deriving the proof.

\section{The lower bound}

Before proving the upper bound, we give a construction to show  that $N(n) \geq \NF(n)$.
This construction is essentially the same as the ones in \cite{NL,VGL}.

\begin{theorem}\label{thm:construction}
$N(n) \geq \NF(n)$.
\end{theorem}
\begin{proof}
First, we show that {$N(3t+1) \geq 2t+1$}:
\begin{enumerate}
\item Place a dot in the {leftmost cell} of the {$(2t+1)$st row} (where we number rows
from top to bottom).
\item Place $t$ more dots, each two squares to the right and one square up from the previous dot.
\item Place a dot in the {$(t+2)$nd cell from the left} in the {bottom row}.
\item Place $t-1$ more dots, each two squares to the right and one square up from the previous dot.
\end{enumerate}
It is easily verified that at most one dot is contained in each row, column, or 
standard diagonal.

Next, {$N(3t+2) \geq N(3t+1) \geq 2t+1$} (it suffices to add a row of empty cells).
Finally, {$N(3t) \geq N(3t+1) - 1 \geq 2t$} (delete the bottom row
of cells from a triangle of side $3t+1$, noting that any row contains at most one dot).
\end{proof}

\begin{example}{\rm
We show in Figure \ref{fig:n=7} that $N(7) \geq 5$ by applying the construction given in
Theorem \ref{thm:construction}.}
\end{example}

\section{A new proof of the upper bound}

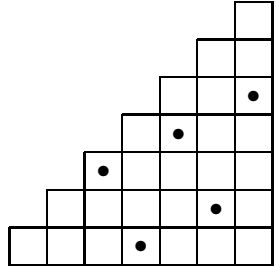
\begin{figure}[tb]
\caption{$N(7) \geq 5$}
\label{fig:n=7}

\begin{equation*}
\setlength{\unitlength}{.5mm}
\begin{picture}(90,90)
\put(10,10){\line(1,0){70}}
\put(10,20){\line(1,0){70}}
\put(20,30){\line(1,0){60}}
\put(30,40){\line(1,0){50}}
\put(40,50){\line(1,0){40}}
\put(50,60){\line(1,0){30}}
\put(60,70){\line(1,0){20}}
\put(70,80){\line(1,0){10}}

\put(10,10){\line(0,1){10}}
\put(20,10){\line(0,1){20}}
\put(30,10){\line(0,1){30}}
\put(40,10){\line(0,1){40}}
\put(50,10){\line(0,1){50}}
\put(60,10){\line(0,1){60}}
\put(70,10){\line(0,1){70}}
\put(80,10){\line(0,1){70}}

\put(35,35){\circle*{3}}
\put(55,45){\circle*{3}}
\put(75,55){\circle*{3}}
\put(45,15){\circle*{3}}
\put(65,25){\circle*{3}}

\end{picture}
\end{equation*}
\end{figure}

The computation of $N(n)$ can be formulated as an 
{integer linear program}. Suppose we number the cells
as indicated in the Figure \ref{6.fig} (where $n=6$):
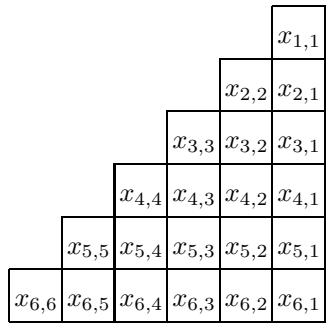
\begin{figure}
\caption{Labelling the cells in a triangle of size $6$}
\label{6.fig}
\begin{equation*}
\setlength{\unitlength}{.7mm}
\begin{picture}(75,70)
\put(10,10){\line(1,0){60}}
\put(10,20){\line(1,0){60}}
\put(20,30){\line(1,0){50}}
\put(30,40){\line(1,0){40}}
\put(40,50){\line(1,0){30}}
\put(50,60){\line(1,0){20}}
\put(60,70){\line(1,0){10}}

\put(10,10){\line(0,1){10}}
\put(20,10){\line(0,1){20}}
\put(30,10){\line(0,1){30}}
\put(40,10){\line(0,1){40}}
\put(50,10){\line(0,1){50}}
\put(60,10){\line(0,1){60}}
\put(70,10){\line(0,1){60}}

\put(11,13){{$x_{6,6}$}}
\put(21,13){{$x_{6,5}$}}
\put(31,13){{$x_{6,4}$}}
\put(41,13){{$x_{6,3}$}}
\put(51,13){{$x_{6,2}$}}
\put(61,13){{$x_{6,1}$}}

\put(21,23){{$x_{5,5}$}}
\put(31,23){{$x_{5,4}$}}
\put(41,23){{$x_{5,3}$}}
\put(51,23){{$x_{5,2}$}}
\put(61,23){{$x_{5,1}$}}

\put(31,33){{$x_{4,4}$}}
\put(41,33){{$x_{4,3}$}}
\put(51,33){{$x_{4,2}$}}
\put(61,33){{$x_{4,1}$}}

\put(41,43){{$x_{3,3}$}}
\put(51,43){{$x_{3,2}$}}
\put(61,43){{$x_{3,1}$}}

\put(51,53){{$x_{2,2}$}}
\put(61,53){{$x_{2,1}$}}

\put(61,63){{$x_{1,1}$}}

\end{picture}
\end{equation*}
\end{figure}

Define $x_{i,j} = 1$ if the corresponding cell contains a dot;
define $x_{i,j} = 0$ otherwise.
The sum of the variables in each row, column, and standard diagonal is at most 1.  
This leads to   {constraints} of the form
\begin{equation*}
\sum_{j=1}^{i}x_{i,j}\leq 1, \quad \quad \mbox{for } i=1,2,\dotsc,n 
\end{equation*}

\begin{equation*}
\sum_{i=j}^{n}x_{i,j}\leq 1, \quad \quad \mbox{for } j=1,2,\dotsc,n 
\end{equation*}
and
\begin{equation*}
\sum_{i=k+1}^{n}x_{i,i-k}\leq 1, \quad \quad \mbox{for } k=0,1,\dotsc,n-1.
\end{equation*}
Finally, $x_{i,j} \in \{0,1\}$ for all $i,j$.
The {objective function} is to maximize $\sum x_{i,j}$ subject to the 
above constraints.
%\medskip
It is obvious that the optimal solution to this integer program is
$N(n)$.

It is possible to relax the integer program to obtain a linear program,  replacing the
condition $x_{i,j} \in \{0,1\}$ by 
$0 \leq x_{i,j} \leq 1$ for all $i,j$. In fact, we do not have to 
specify $x_{i,j} \leq 1$ as an explicit constraint since it is already implied by the other constraints; it suffices to require
$0 \leq x_{i,j}$ for all $i,j$. 
Denoting the optimal solution to this linear program
by $\LP(n)$, we have that $\LP(n) \geq N(n)$.

\begin{example}
{\rm
Using Maple, it can be seen that $\LP(6) = 4\frac{2}{7}$ (a solution to the LP having this
value is presented in Figure \ref{LP:n=6}).}
\end{example}

\begin{figure}
\caption{The optimal solution  to $LP(6)$}
\label{LP:n=6}
\begin{equation*}
\setlength{\unitlength}{.7mm}
\begin{picture}(90,80)
\put(10,10){\line(1,0){60}}
\put(10,20){\line(1,0){60}}
\put(20,30){\line(1,0){50}}
\put(30,40){\line(1,0){40}}
\put(40,50){\line(1,0){30}}
\put(50,60){\line(1,0){20}}
\put(60,70){\line(1,0){10}}

\put(10,10){\line(0,1){10}}
\put(20,10){\line(0,1){20}}
\put(30,10){\line(0,1){30}}
\put(40,10){\line(0,1){40}}
\put(50,10){\line(0,1){50}}
\put(60,10){\line(0,1){60}}
\put(70,10){\line(0,1){60}}

\put(14,13){{$0$}}
\put(24,13){{$0$}}
\put(34,13){{$\frac{2}{7}$}}
\put(44,13){{$\frac{4}{7}$}}
\put(54,13){{$\frac{1}{7}$}}
\put(64,13){{$0$}}

\put(24,23){{$\frac{2}{7}$}}
\put(34,23){{$0$}}
\put(44,23){{$\frac{3}{7}$}}
\put(54,23){{$\frac{1}{7}$}}
\put(64,23){{$\frac{1}{7}$}}

\put(34,33){{$\frac{5}{7}$}}
\put(44,33){{$0$}}
\put(54,33){{$0$}}
\put(64,33){{$\frac{2}{7}$}}

\put(44,43){{$0$}}
\put(54,43){{$\frac{5}{7}$}}
\put(64,43){{$\frac{2}{7}$}}

\put(54,53){{$0$}}
\put(64,53){{$\frac{2}{7}$}}

\put(64,63){{$0$}}

\end{picture}
\end{equation*}

\end{figure}

Next, we tabulate some solutions to $\LP(n)$ for small values of $n$
 in Table \ref{tab1}. Based on the numerical data in Table \ref{tab1}, it is natural to formulate 
a conjecture about  $\LP(n)$:
 
\begin{table}[tb]
\caption{Optimal solutions to the integer and linear programs for small $n$}
\label{tab1}
\[
\renewcommand{\arraystretch}{1.35}
\begin{array}{c|c|c|c}
n & N(n) & \LP(n) & \LP(n) - N(n)\\ \hline
3 & 2 & 2\frac{1}{4} & \frac{1}{4} \\ \hline
4 & 3 & 3 & 0 \\ \hline
5 & 3 & 3\frac{3}{5} & \frac{3}{5} \\\hline
6 & 4 & 4\frac{2}{7} & \frac{2}{7} \\\hline
7 & 5 & 5 & 0 \\\hline
8 & 5 & 5\frac{5}{8} & \frac{5}{8} \\\hline
9 & 6 & 6\frac{3}{10} & \frac{3}{10} \\\hline
10 & 7 & 7 & 0 \\\hline
11 & 7 & 7\frac{7}{11} & \frac{7}{11} \\\hline
12 & 8 & 8\frac{4}{13} & \frac{4}{13} 
\end{array}
\]
\end{table}

\begin{conjecture}[LP Conjecture] 
Define
\begin{eqnarray*}
\LPF(3t) &=& 2t + \frac{t}{3t+1} \\
\LPF(3t+1) &=& 2t+1 \\
\LPF(3t+2) &=& 2t+1 + \frac{2t+1}{3t+2}.
\end{eqnarray*} Then we conjecture that $\LP(n) = \LPF(n)$.
\end{conjecture}

It is easy to show the following:
\begin{theorem}
\label{LPproof.thm}
{If the LP Conjecture is true, then 
$N(n) = \NF(n)$.}
\end{theorem} 

\begin{proof}
First, the LP Conjecture asserts that 
\begin{equation}
\label{eq1} \LP(n) = \LPF(n).
\end{equation}
 Because $N(n)$ is an integer and $N(n) \leq \LP(n)$, 
we have that \begin{equation}
\label{eq2} N(n) \leq \lfloor \LP(n) \rfloor . \end{equation}
Simple arithmetic establishes that 
\begin{equation}
\label{eq3} \lfloor \LPF(n) \rfloor = \NF(n).\end{equation}
Combining (\ref{eq1}), (\ref{eq2}) and (\ref{eq3}), we have
\[N(n) \leq \lfloor \LP(n) \rfloor = \lfloor \LPF(n) \rfloor = \NF(n) .\]
We showed in Theorem \ref{thm:construction}
that $N(n) \geq \NF(n)$; hence $N(n) = \NF(n).$
\end{proof}

The optimal solution to the linear program for $n=6$ that we presented in
Figure \ref{LP:n=6} does not seem to have much apparent structure
that could be the basis of a mathematical proof. Indeed, most of
the small optimal solutions that we obtained are quite irregular,
which suggests that 
proving the LP conjecture could
be difficult.
We circumvent this problem by instead studying the dual LP
and appealing to weak duality.

An LP in {\it standard form} is specified as:

\begin{center}
\begin{tabular}{|ll|}\hline
maximize  & $c^T x$\rule{0in}{.18in}  \\
subject to & $Ax \leq b$, $x \geq 0$.\\ \hline
\end{tabular}
\end{center}
This is often called the {\it primal LP}.
Any vector $x$ such that $Ax \leq b$, $x \geq 0$ is called a 
{\it feasible solution}. The {\it objective function}
is the value to be maximized, namely, $c^T x$.

The corresponding {\it dual LP} is specified as:
\begin{center}
\begin{tabular}{|ll|}\hline
minimize  & $b^T y$\rule{0in}{.18in}  \\
subject to & $A^T y \geq c$, $y \geq 0$.\\ \hline
\end{tabular}
\end{center}
Here, a feasible solution is any vector $y$ such that 
$A^T y \geq c$, $y \geq 0$. The objective function
is $b^T y$.

We will use the following classic theorem.

\begin{theorem}[Weak Duality Theorem]
\label{weakduality.thm}
The objective function value of the dual LP at any feasible solution is always greater than or equal to the objective function value of the primal LP at any feasible solution.
\end{theorem}

We now describe the dual LP for our problem. 
Suppose we label the rows of our triangle by $r_1,r_2,\dots,r_{n}$, such that 
{$r_i$ is the row containing $i$ squares}, 
and we label the columns and diagonals similarly.  
The following simple lemma is very useful. 

\begin{lemma}
\label{sum.lem}
If a cell is in row $r_i$, column $c_j$ and diagonal $d_k$, then
 {$i+j+k=2n+1$}.
\end{lemma}

In fact, it is not hard to see that 
there is a {bijection} from the set of $n(n+1)/2$ cells
to the set of triples
\[ \mathcal{T} = \{ (i,j,k) : i+j+k = 2n+1, 1 \leq i,j,k\leq n\} .\]

In the dual LP, the {variables} are $r_1,r_2,\dots,r_{n}$,  
 $c_1,c_2,\dots,c_{n}$, $d_1,d_2,\dots, d_{n}$.
There is a {constraint} for 
each cell $C$. If $C$ is in row $r_i$, column $c_j$ and diagonal $d_k$,
then the corresponding constraint is
\[ r_i + c_j +d_k \geq 1.\]
 The {objective function} is to minimize $\sum r_i + \sum c_j + \sum d_k$.

\medskip

It turns out that there exist optimal solutions for the dual LP that have a very {simple, regular} structure. These were obtained by extrapolating solutions for small cases found by Maple.

When $n = 3t+1$, define 
\begin{equation}
\label{1.soln}
r_i = c_i =\max\left\{0,\frac{i-t-1}{3t+1}\right\},\quad
d_i=\max\left\{0,\frac{i-t}{3t+1}\right\}.
\end{equation}

When $n = 3t+2$, define
\begin{equation}
\label{2.soln}
r_i= c_i = d_i =\max\left\{0,\frac{i-t-1}{3t+2}\right\}.
\end{equation}

When $n = 3t$, define
\begin{equation}
\label{0.soln}
r_i = c_i = d_i =\max\left\{0,\frac{i-t}{3t+1}\right\}.
\end{equation}

\begin{lemma}
\label{feas.lem}
The values $r_i, c_i$ and $d_i$ defined in 
\eqref{1.soln}, \eqref{2.soln} and \eqref{0.soln}
are {feasible} for the dual LP, and
the {value of the objective function} for the dual LP at these solutions is
$\LPF(n)$.
\end{lemma}

\begin{proof}
First we consider the case $n = 3t+1$.  Consider any cell $C$, and suppose $C$ is in row $r_i$, column $c_j$ and diagonal $d_k$.
Then 
we have  that
\begin{align*}
r_i+c_j+d_{k}&\geq\frac{i-t-1}{3t+1}+\frac{j-t-1}{3t+1}+\frac{k-t}{3t+1}\\
&=\frac{i+j+k- (3t+2)}{3t+1} \\
&=\frac{6t+3 - (3t+2)}{3t+1} \quad \text{(applying Lemma \ref{sum.lem})}\\
&=1.
\end{align*}
Therefore all constraints are satisfied.  The value of the objective function is 
\begin{align*}
&\frac{1}{3t+1}\left(%
\sum_{i=t+1}^{3t+1}(i-t-1)+
\sum_{i=t+1}^{3t+1}(i-t-1)+
\sum_{i=t}^{3t+1}(i-t)
\right)\\
&= \frac{1}{3t+1}\left(\frac{2t(2t+1)}{2}+\frac{2t(2t+1)}{2}+
\frac{(2t+1)(2t+2)}{2}\right)\\
&= \frac{(2t+1)(3t+1)}{3t+1}\\
&= 2t+1\\
&= \LPF(3t+1).
\end{align*}

The proofs for $n = 3t+2$ and $n= 3t$ are very similar. 
\end{proof}

Our new proof of Theorem 
\ref{main.thm} follows immediately from Lemma \ref{feas.lem}
by slightly modifying the proof of
Theorem \ref{LPproof.thm}.

\begin{proof}
First, from weak duality and Lemma \ref{feas.lem}, we have
\begin{equation*}
\LP(n) \leq \LPF(n).
\end{equation*}
 Combining this inequality with (\ref{eq2}) and (\ref{eq3}), we have
\[N(n) \leq \lfloor \LP(n) \rfloor \leq \lfloor \LPF(n) \rfloor = \NF(n) .\]
We showed in Theorem \ref{thm:construction}
that $N(n) \geq \NF(n)$; hence $N(n) = \NF(n).$
\end{proof}

\section{Discussion}

We investigated 
the ``dots in triangles problem" due to an application 
to honeycomb arrays (see \cite{BPPS}).
However, we did not
realize that the dots in triangles problem had already been solved. 
Since we did not know the
value of $N(n)$, we adopted an ``experimental" approach:
\begin{enumerate}
\item 
We used Maple to gather some numerical data.
\item We formulated (obvious) conjectures
based on the numerical data.
\item Finally, we proved the conjectures mathematically. 
\end{enumerate}
Many problems in combinatorics are amenable to such
an approach, but this particular problem serves as an ideal
illustration of the usefulness of
this methodology. Indeed, the problem seemed almost
to ``solve itself'', with minimal thought or human ingenuity required! 

\medskip

It should also be emphasized that, in the end, the resulting proof is quite short and simple:
\begin{enumerate}
\item By a suitable direct construction, prove that 
$N(n) \geq \left\lfloor \frac{2n+1}{3} \right\rfloor $.
\item Show that the dual LP has a feasible solution whose objective
function value is {less than} 
$\left\lfloor \frac{2n+1}{3} \right\rfloor + 1$.
\end{enumerate}

The first conjecture we posed was the LP Conjecture, concerning the 
optimal solutions to the LP. 
In general, to prove a feasible solution to an LP is optimal, 
it is necessary to do the following:
\begin{enumerate}
\item Find a feasible solution to the primal LP and denote the value of the  objective function by $C$.
\item Find a feasible solution to the dual LP and denote the value of the  objective function by $C^*$.
\end{enumerate}
If $C = C^*$, then the solution to the LP is optimal (this is often called
{\it strong duality}).

When $n \equiv 1 \bmod 3$, our work in fact proves the LP conjecture.
This is because Theorem \ref{thm:construction} yields a solution to the 
primal LP whose objective function value matches the solution we later found to the dual LP.
However, when $n \not\equiv 1 \bmod 3$, we do not have a general solution to the primal LP 
whose objective function value matches the solution to the dual LP.
Although we are confident that the LP conjecture is
also true for these values of $n$, proving it could get messy!

\section*{Acknowledgements}

Research of SRB and MBP was supported by EPSRC grant EP/D053285/1 and
research of DRS was supported by NSERC grant 203114-06. The authors
would like to thank Tuvi Etzion for discussions, funded by a Royal
Society International Travel Grant, which inspired this line of
research. Many thanks also to Bill Martin, for comments on an earlier
version of this paper.

\end{document}